\documentclass[twocolumn,pra,letterpaper]{revtex4}
\usepackage{hyperref}
\usepackage{latexsym,amsmath,amssymb,amsfonts,graphicx,color,amsthm}
\usepackage{enumerate}

\newcommand{\cB}{\mathcal{B}}
\newcommand{\cC}{\mathcal{C}}

\newcommand{\cE}{\mathcal{E}}
\newcommand{\cF}{\mathcal{F}}

\newcommand{\cH}{\mathcal{H}}

\newcommand{\cK}{\mathcal{K}}

\newcommand{\cN}{\mathcal{N}}

\newcommand{\cP}{\mathcal{P}}

\newcommand{\cR}{\mathcal{R}}
\newcommand{\cS}{\mathcal{S}}

\newcommand{\cV}{\mathcal{V}}

\newcommand{\tr}{\text{tr}}

\newcommand{\RP}{\cR_P}

\newcommand{\etaop}{\eta_\text{op}}

\newcommand{\Rop}{\cR_\text{op}}

\newtheorem{theorem}{Theorem}
\newtheorem{lemma}[theorem]{Lemma}
\newtheorem{corollary}[theorem]{Corollary}

\newtheorem{condition}{Condition}

\begin{document}

\title{Towards a Unified Framework for Approximate Quantum Error Correction}
\author{Prabha Mandayam}
\affiliation{The Institute of Mathematical Sciences, Taramani, Chennai - 600113, India.}

\author{Hui Khoon Ng}
\affiliation{Centre for Quantum Technologies, National University of Singapore, Block S15, 3 Science Drive 2, Singapore 117543, Singapore,\\and DSO National Laboratories, 20 Science Park Drive, Singapore 118230, Singapore.}
\date{\today}

\begin{abstract}
Recent work on approximate quantum error correction (QEC) has opened up the possibility of constructing subspace codes that protect information with high fidelity in scenarios where perfect error correction is impossible. Motivated by this, we investigate the problem of approximate \emph{subsystem} codes. Subsystem codes extend the standard formalism of subspace QEC to codes in which only a subsystem within a subspace of states is used to store information in a noise-resilient fashion.
Here, we demonstrate easily checkable sufficient conditions for the existence of approximate subsystem codes. Furthermore, for certain classes of subsystem codes and noise processes,
we prove the efficacy of the transpose channel as a simple-to-construct recovery map that works nearly as well as the optimal recovery channel. This work generalizes our earlier approach [Ng, Mandayam, Phys. Rev. A \textbf{81} (2010) 62342] of using the transpose channel for approximate correction of subspace codes to the case of subsystem codes, and brings us closer to a unifying framework for approximate QEC.
\end{abstract}

\maketitle

\section{Introduction}
The textbook paradigm of quantum error correction (QEC) focuses on the case of \emph{perfect} error correction \cite{EM96, BDSW96}, where the code $\cC$ and the noise are such that there exists a recovery operation that completely removes the effects of the noise on the information stored in the code.
Mathematically, this idea is captured by a set of conditions for perfect error correction~\cite{KnillLaflamme} that must be satisfied by the code as well as the noise process.

That such perfect QEC conditions can be satisfied tends to be special rather than generic.
The prototypical example is that of independent noise acting on a few physical qubits, and one finds codes that satisfy the perfect QEC conditions assuming that no more than $t$ of the qubits have errors.
In such a scenario, what is taken as the noise process $\cE$ in the QEC conditions is the part of the noise that describes $t$ or fewer errors, while the full physical noise process $\cE_0$ contains terms describing more than $t$ errors, albeit with a lower probability of occurrence.
A code that satisfies the QEC conditions for $\cE$ will thus only satisfy the conditions \emph{approximately} for the full noise process $\cE_0$.
Furthermore, in practice, it is unrealistic to expect complete characterization of the noise process.
Thus, a code designed to satisfy the perfect QEC conditions for the \emph{expected} noise process will typically only satisfy those conditions approximately for the true noise process.
This motivates the idea of \emph{approximate} quantum error correction (AQEC), where the recovery operation removes most, but not necessarily all, the effects of noise on the information stored in the code.

Recent studies on AQEC, using analytical \cite{Leung,BK,Tyson,BenyOreshkov1,aqecPRA,Renes} and numerical~\cite{Yamamoto,Reimpell,Kosut08,Fletcher07b} approaches, have discovered examples of approximate codes that allow for recovery of stored information with fidelity comparable to that of perfect QEC codes, while making use of fewer physical resources. These results suggest that the requirement for perfect recovery may be too stringent for certain tasks and approximate QEC may be more natural and practical.

In~\cite{aqecPRA}, we demonstrated a universal, near-optimal recovery map---the \emph{transpose channel}~\cite{BK,petzbook}---for AQEC codes with a subspace structure, wherein information is stored in an entire subspace of the Hilbert space of the physical quantum system. Optimality was defined in terms of the worst-case recovery fidelity over all states in the code. Our analytical approach was a departure from earlier work relying on exhaustive numerical search for the optimal recovery map, with optimality based on entanglement fidelity~\cite{Reimpell,FletcherThesis,Kosut08}. We obtained quantitative bounds showing the efficacy of the transpose channel as a universal and analytical recovery operation that works well, regardless of the noise process or the code used.
This allowed for complete characterization of approximate subspace codes, in terms of necessary and sufficient conditions for approximate correctability, and provided an easy route for constructing approximate subspace codes.

In this article, we extend our approach based on the transpose channel to the more general case of AQEC codes with a \emph{subsystem} structure, where information is stored only in a subsystem of the code subspace.
A subsystem code (sometimes referred to as an \emph{operator} QEC code) has a bipartite tensor-product structure, where one subsystem $A$ (the \emph{correctable subsystem}) is correctable under action of the noise, while the other subsystem $B$ (the \emph{noisy subsystem}) can be disturbed by the noise beyond repair~\cite{OQECC, OQEClong, NielsenPoulin07}.
The information to be protected against noise is stored only in subsystem $A$.
Subspace codes can be viewed as special cases of subsystem codes with a trivial noisy subsystem.
While this generalization does not lead to new families of codes, the alternate perspective does sometimes lead to more efficient decoding procedures~\cite{Poulin05stabilizer,bacon06oqec}, and hence to better fault-tolerant schemes and improved bounds on the accuracy threshold~\cite{aliferis07}.
Starting with the Bacon-Shor codes~\cite{bacon06oqec}---a family of subsystem codes arising from Shor's 9-qubit code---several examples of perfectly correctable stabilizer subsystem codes have been constructed~\cite{aly06subsystemcodes}.
Here, we address the general question of characterizing approximate subsystem codes, and explore the extent to which the transpose channel is useful towards understanding approximate codes.

After a preliminary section setting down basic definitions and notations, we begin by proving a set of perfect QEC conditions (Section \ref{sec:perfectOQEC}) that is completely equivalent to the standard well-known QEC conditions. This alternate set of conditions clarifies the role of the transpose channel in perfect QEC. Furthermore, it serves as a natural starting point for perturbation to a set of sufficient conditions for approximate QEC (Section~\ref{sec:sufficiency}). We then proceed, in Section \ref{sec:necessary}, to show the near-optimality of the transpose channel recovery map for AQEC for  four classes of codes and noise processes. These four classes provide evidence towards our conjecture that the transpose channel is near-optimal for arbitrary subsystem codes and noise, which, if true, would establish a simple, analytical, and universal framework for the study of approximate codes. We conclude with a few suggestions for future work.

\section{Basic definitions}\label{sec:definition}

We consider a decomposition of the Hilbert space of our quantum system,
\begin{equation}\label{eq:decompH}
\cH=\cH_A\otimes \cH_B+\cK.
\end{equation}
Suppose we wish to store information in the $\cH_A$ factor.
$\cH_{AB}\equiv\cH_A\otimes \cH_B$ is the Hilbert space of a composite system comprising two subsystems $A$ and $B$ of dimensions $d_A$ and $d_B$ respectively.
We denote the projector onto $\cH_{AB}$ as $P$.
$P$ can also be written as a tensor-product: $P=P_A\otimes P_B$, where $P_{A(B)}$ is the projector onto $\cH_{A(B)}$.
In principle, subsystems $A$ and $B$ may only correspond to mathematical tensor-product factors in the decomposition, rather than ``natural" separate physical degrees of freedom of the quantum system.
In practice, one might prefer to work with $A$ and $B$ that are natural degrees of freedom for easy experimental accessibility.
Also, it is often helpful to use a decomposition of $\cH$ that is not arbitrarily invented by the experimenter, but induced by the structure of the noise afflicting the quantum system, so as to identify a subsystem that best ensures survival of the stored information.

Information is stored as a choice between states of subsystem $A$.
The state on subsystem $B$ can be arbitrary and carries no information. More concretely, we make use of a \emph{code} $\cC$, comprising all product states on $AB$,
\begin{equation}\label{eq:code}
\cC\equiv\{\rho=\rho_A\otimes\rho_B,~\forall\rho_A\in\cS(\cH_A),\rho_B\in\cS(\cH_B)\},
\end{equation}
where $\cS(\cH_{A(B)})$ denotes the set of all states (density operators) on subsystem $A(B)$.
The information is stored only in subsystem $A$ in that two states $\rho_A\otimes \tau_B$ and $\rho_A\otimes \sigma_B$ differing only in the state of $B$ correspond to the same encoded information.

We wish to examine the longevity of the information stored in subsystem $A$ in the presence of noise.
We describe the noise by a quantum channel acting on $AB$, that is, a completely positive (CP), trace-preserving (TP) map $\cE:\cB(\cH_{AB})\longrightarrow \cB(\cP_\cE)$.
Here, $\cB(\cV)$ refers to the set of all bounded operators on a vector space $\cV$. $\cP_\cE$ is the support of $\cE(\cB(\cH_{AB}))$, or equivalently, the support of $\cE(P)$.
$\cE$ can be specified by a set of Kraus operators $\{E_i\}_{i=1}^N$, so that $\cE$ acts as
\begin{equation}
\cE(\rho)=\sum_{i=1}^NE_i\rho E_i^\dagger.
\end{equation}
That $\cE$ is TP translates into the statement $\sum_{i=1}^NE_i^\dagger E_i=P$.
The Kraus representation of a CPTP channel is non-unique: If $\{E_i\}$ is a Kraus representation of $\cE$, then $\{F_j\equiv \sum_i u_{ij}E_i\}$, for unitary $(u_{ij})$, is a Kraus representation of the same channel.
A recovery operation $\cR:\cB(\cP_\cE)\longrightarrow \cB(\cH_{AB})$ performed
after each application of the noise $\cE$, to attempt to reverse the effects of the noise, is also described as a CPTP map.

Since information is stored in subsystem $A$ only, we are concerned only with how well the noise preserves the information initially stored in $A$, while any state on $B$ can be distorted beyond repair by the noise.
Heuristically, we say that a code $\cC$ is \emph{approximately correctable} under noise $\cE$ if and only if there exists a CPTP recovery map $\cR$ such that
\begin{equation}
\tr_B[(\cR\circ\cE)(\rho)]\simeq\tr_B(\rho)\quad\forall\rho\in\cC,
\end{equation}
where $\tr_B(\cdot)$ denotes the partial trace over subsystem $B$.

This heuristic notion is formalized by quantifying the deviation of the recovered state from the initial state in terms of the fidelity between the two states.
The fidelity between two states $\rho$ and $\sigma$ is $F(\rho,\sigma)\equiv \tr\sqrt{\rho^{1/2}\sigma\rho^{1/2}}$. For the case of $\rho$ being a pure state $\psi\equiv |\psi\rangle\langle\psi|$, $F$ can be written as
\begin{equation}
F(|\psi\rangle,\sigma)\equiv \sqrt{\langle \psi|\sigma|\psi\rangle}.
\end{equation}
We define the \emph{fidelity loss for state $\rho$}, $\eta_\cR\{\rho\}$, under noise $\cE$ and recovery $\cR$, as the deviation from 1 of the square of the fidelity between the initial state and the recovered state, that is,
\begin{equation}
\eta_\cR\{\rho\}\equiv 1-F^2{\big(\tr_B(\rho),\tr_B[(\cR\circ\cE)(\rho)]\big)}.
\end{equation}
The performance of a recovery $\cR$ on a code $\cC$ is then characterized by the \emph{fidelity loss for $\cC$} defined as
\begin{equation}\label{eq:fidelityloss}
\eta_\cR\{\cC\}\equiv \max_{\rho\in\cC}\eta_\cR\{\rho\}.
\end{equation}
How well $\cR$ recovers the information initially stored in subsystem $A$ is hence gauged by the \emph{worst-case fidelity} (over all states in the code) between the initial and recovered states. Because the fidelity is jointly concave in its arguments, the worst-case fidelity is always attained on a pure state on $AB$. The maximization in Eq.~\eqref{eq:fidelityloss} can thus be restricted to pure states on $AB$ only. Often, when the meaning is clear from the context, we will drop the argument from $\eta_\cR\{\cC\}$ and simply write $\eta_\cR$.

Let $\Rop$ be the recovery map with the smallest fidelity loss among all possible recovery maps for code $\cC$, that is,
\begin{equation}
\etaop\{\cC\}\equiv\eta_{\Rop}\{\cC\}=\min_\cR\eta_\cR\{\cC\}.
\end{equation}
We refer to $\Rop$ as the \emph{optimal recovery}, and $\etaop$ as the \emph{optimal fidelity loss}.
As is clear from the notation, whether or not a recovery map is optimal for a given noise process depends on the code in question.

A code $\cC$ with $\etaop=0$ under noise $\cE$ is said to be \emph{perfectly correctable} on $A$ under $\cE$.
In general, we say that a code is \emph{$\epsilon$-correctable} on $A$ under noise $\cE$ if $\epsilon\geq\etaop$, which means that it is possible to recover the information stored in $A$ with a fidelity no smaller than $\sqrt{1-\epsilon}$.
Approximate \emph{noiseless subsystems} are included within this framework by considering codes for which the identity map (the ``do nothing" operation) is sufficient as an approximate recovery for the code.

Central to our analysis is a recovery map built from the noise channel and code, known as the \emph{transpose channel} (see previous uses of this channel in Refs.~\cite{petz2003,haydenpetz,BK,IPSPaper,IPSPaper2}).
The transpose channel corresponding to a channel $\cE$ and code $\cC$, denoted as $\RP$, is defined in a manifestly representation-invariant way as
\begin{equation}
\RP\equiv \cP_\cC\circ\cE^\dagger\circ\cN.
\end{equation}
Here, $\cE^\dagger$ is the adjoint of $\cE$, that is, the channel with Kraus operators $\{E_i^\dagger\}_{i=1}^N$ if $\cE$ has Kraus operators $\{E_i\}_{i=1}^N$. $\cN$ is a normalization map $\cN(\cdot)\equiv \cE(P)^{-1/2}(\cdot)\cE(P)^{-1/2}$ (the inverse is taken on the support of $\cE(P)$). $\cP_\cC$ is the projection onto the support of $\cC$, $\cP_\cC(\cdot)=P(\cdot)P$. One can write $\RP$ explicitly in terms of its Kraus operators $\{R_i^P\}_{i=1}^N$, where
\begin{equation}
R_{i}^{P} \equiv PE_{i}^\dagger \cE(P)^{-1/2}.
\end{equation}
$\RP$ is trace-preserving (TP) on $\cP_\cE$. The fidelity loss obtained using the transpose channel as the recovery is denoted by $\eta_P$.

\section{Approximate QEC: sufficient conditions}\label{sec:AQECSuff}

\subsection{Perfect QEC conditions}\label{sec:perfectOQEC}

We begin with the case of perfect error correction, where, there exists a recovery map such that the fidelity of any state on $A$ after noise and recovery attains the maximal value of 1. Necessary and sufficient algebraic conditions for the existence of a perfectly correctable code for a given channel $\cE$ are expressed by the following theorem:
\begin{theorem}\label{thm:equivcond}
Consider a CPTP noise process $\cE$ with Kraus representation $\{E_i\}$ acting on $AB$, and a code $C$ on $\cH_{AB}$ as defined in Eq.~\eqref{eq:code}. $\cC$ is perfectly correctable on $A$ under $\cE$ if and only if
\begin{equation}\label{cond1}
PE_i^\dagger \cE(P)^{-1/2}E_jP=P_A\otimes B_{ij},
\end{equation}
for all $i,j$, and $B_{ij}\in\cB(\cH_B)$.
\end{theorem}
\noindent The special case of Theorem~\ref{thm:equivcond} for subspace codes appeared in \cite{aqecPRA}, and the proof of this generalization (provided in Appendix \ref{app:perfectCond}) follows a similar logic.

Algebraic conditions for perfect error correction for subsystem codes were originally discovered in \cite{NielsenPoulin07, OQECC, OQEClong}, generalizing the well-known perfect QEC conditions for subspace codes \cite{KnillLaflamme}.
Compared to the original QEC conditions, our conditions given above differ only in the appearance of the $\cE(P)^{-1/2}$ factor on the left-side of Eq.~\eqref{cond1}. However, this alternate form of the conditions offers better intuition on the correctability of codes.
Observe that the expression on the left-side of Eq.~\eqref{cond1} is a Kraus operator $R_i^PE_j$ of the channel $\RP\circ\cE$, from which we can immediately conclude that $\tr_B\{(\RP\circ\cE)(\rho)\}=\tr_B(\rho)$ for any $\rho\in\cC$, as is required for perfect correctability on $A$.
Theorem \ref{thm:equivcond} can thus be viewed as demonstrating correctability of codes by explicitly giving the recovery map---the transpose channel $\RP$---needed to perfectly recover the state on subsystem $A$ after the action of the channel $\cE$.

\subsection{Sufficient AQEC conditions}\label{sec:sufficiency}

The form of the QEC conditions given in Eq.~\eqref{cond1} is particularly well-suited for perturbation to approximate QEC, as was previously pointed out for the special case of subspace codes in \cite{aqecPRA}.
Theorem \ref{thm:equivcond} states that $\RP\circ\cE$ acts as the identity channel on subsystem $A$.
Perturbing Eq.~\eqref{cond1}, by adding to the right-side a small correction to $P_A\otimes B_{ij}$, modifies this to the statement that $\RP\circ\cE$ acts \emph{nearly} as the identity channel on subsystem $A$.
This provides a natural route to sufficient conditions for approximate subsystem codes:
If the perturbation to Eq.~\eqref{cond1} is small enough, the code is $\epsilon$-correctable on $A$ with small $\epsilon$. What remains is to relate quantitatively the size of the perturbation to $\epsilon$, which is the content of the following theorem:

\begin{theorem}\label{thm:suff}
Consider a CPTP noise channel $\cE$ with Kraus representation $\{E_{i}\}$ and a code $\cC$ on $\cH_{AB}$ as defined in Eq.~\eqref{eq:code}. Suppose
\begin{equation}\label{eq:AOQECcond}
PE_{i}^\dagger \cE(P)^{-1/2} E_{j}P = P_{A}\otimes B_{ij} + \Delta_{ij},
\end{equation}
for all $i,j$, $B_{ij} \in\cB(\cH_B)$, and $\Delta_{ij}\in\cB(\cH_{AB})$. Then, $\cC$ is $\epsilon$-correctable on $A$ under $\cE$ for $\epsilon\geq\eta_P$, where
\begin{align}\label{eq:etaP}
\eta_P&\equiv \max_{|\psi_A,\phi_B\rangle}\langle \phi_B|\sum_{ij}{\left[\langle\psi_A|\Delta_{ij}^\dagger\Delta_{ij}|\psi_A\rangle\right.}\\
&\hspace{2.65cm}{\left.-\langle\psi_A|\Delta_{ij}^\dagger|\psi_A\rangle\langle\psi_A|\Delta_{ij}|\psi_A\rangle\right]}|\phi_B\rangle.\nonumber
\end{align}
\end{theorem}

\begin{proof}
The TP condition on $\RP\circ\cE$ gives the relation $P=\sum_{ij}[P_A\otimes B_{ij}^\dagger B_{ij}+\Delta_{ij}^\dagger\Delta_{ij}+(P_A\otimes B_{ij}^\dagger)\Delta_{ij}+\Delta_{ij}^\dagger(P_A\otimes B_{ij})]$. Using this, direct computation gives
\begin{align}
&\quad F^2{\left[|\psi_A\rangle,(\tr_B\circ\RP\circ\cE)(|\psi_A,\phi_B\rangle\langle\psi_A,\phi_B|)\right]}\\
&=1-\langle\psi_A,\phi_B|\sum_{ij}\Delta_{ij}^\dagger{\left(P_A-|\psi_A\rangle\langle\psi_A|\right)}\otimes P_B\Delta_{ij}|\psi_A,\phi_B\rangle.\nonumber
\end{align}
This yields the expression for $\eta_P$ in Eq.~\eqref{eq:etaP} upon recalling that the worst-case fidelity is attained on a pure state on $AB$.
\end{proof}

While the bound $\epsilon\geq \eta_P$ in Theorem~\ref{thm:suff} is tight, the maximization over all pure product states on $AB$ in the expression for $\eta_P$ may not be easy to evaluate.
Instead, we can relax the bound and obtain a simpler (but weaker) sufficient condition,
\begin{corollary}\label{cor:suff}
$\cC$ is $\epsilon$-correctable on $A$ under $\cE$ if
\begin{equation}
\epsilon\geq \Big\Vert\sum_{ij}\Delta_{ij}^\dagger\Delta_{ij}\Big\Vert,
\end{equation}
where $\Vert\cdot\Vert$ is the operator norm.
\end{corollary}

\begin{proof}
Observe that, for any pure product state $|\psi_A,\phi_B\rangle$, the expression in Eq.~\eqref{eq:etaP} to be maximized is bounded from above by $\langle\psi_A,\phi_B|\sum_{ij}\Delta_{ij}^\dagger\Delta_{ij}|\psi_A,\phi_B\rangle\leq \big\Vert\sum_{ij}\Delta_{ij}^\dagger\Delta_{ij}\big\Vert$. This gives $\eta_P\leq \big\Vert\sum_{ij}\Delta_{ij}^\dagger\Delta_{ij}\big\Vert$, which immediately yields the corollary statement.
\end{proof}
\noindent Corollary \ref{cor:suff} gives an easily checkable sufficient condition, which may be more useful than Theorem \ref{thm:suff} in the search for approximate subsystem codes.

\section{Towards necessary AQEC conditions}\label{sec:necessary}

The previous section discusses sufficient conditions for the existence of approximate QEC codes.
The next natural question to ask is: what about necessary conditions?
For the special case of subspace codes, we obtained necessary conditions by deriving a near-optimality bound for the transpose channel recovery \cite{aqecPRA}.
The near-optimality result led to the conclusion that every approximately correctable subspace code must also be well-corrected by the transpose channel.
This relation to the transpose channel gave rise to necessary conditions for the existence of approximate subspace codes of a form similar to the sufficient conditions.

Extending the near-optimality bound to arbitrary subsystem codes proved difficult.
Nevertheless, as is described in this section, we can show near-optimality of the transpose channel for restricted classes of subsystem codes and noise processes.
More specifically, we consider four scenarios: (A) subspace codes, with trivial subsystem $B$ (a review of results from \cite{aqecPRA}); (B) code states with the maximally mixed state on $B$;  (C) subsystem $B$ is perfectly correctable; and (D) the noise $\cE$ destroys information on $B$.
In each scenario, the transpose channel works nearly as well as the optimal recovery operation, which leads to necessary conditions on the noise process as well as the code, \emph{provided} the restrictions are satisfied.

In the broader picture of arbitrary subsystem codes and noise processes, we believe that the transpose channel still works well whenever the code is approximately correctable.
After all, in the case of perfect QEC, the transpose channel is \emph{the} recovery operation for perfect recovery.
However, the general near-optimality of the transpose channel is only a conjecture at this point.
Here, we seek only to provide evidence towards the conjecture, and leave the proof (or disproof) to future work.

\subsection{Trivial subsystem $B$: subspace codes}\label{sec:subspace}

In \cite{aqecPRA}, the transpose channel was shown to be near-optimal for subspace codes $\cC$, that is, its fidelity loss for code $\cC$ under noise $\cE$ is close to the optimal fidelity loss.
For completeness, we repeat here the quantitative statement of the near-optimality of the transpose channel, adapted to the language suited for this paper:
\begin{theorem}[Corollary 4 of \cite{aqecPRA}]\label{thm:subspace}
Consider a subspace code $\cC$ ($B$ is trivial), with $d_A$ denoting the dimension of $\cH_A$, and optimal fidelity loss $\etaop$ under CPTP noise channel $\cE$. The fidelity loss $\eta_P$ for the transpose channel satisfies
\begin{equation}\label{eq:ineqEta}
\etaop\leq \eta_P\leq\etaop f(\etaop;d_A),
\end{equation}
where $f(\eta;d)$ is the function
\begin{equation}\label{eq:f}
f(\eta;d)\equiv \frac{(d+1)-\eta}{1+(d-1)\eta}=(d+1)+O(\eta).
\end{equation}
\end{theorem}
\noindent The left inequality $\etaop\leq \eta_P$ of Eq.~\eqref{eq:ineqEta} is true simply by definition of $\etaop$. The proof of the right inequality $ \eta_P\leq\etaop f(\etaop;d_A)$ requires the following inequality (derived in \cite{aqecPRA}) which holds for any pure state $\psi_A\equiv |\psi_A\rangle\langle\psi_A|$ in a subspace code $\cC$,
\begin{equation}\label{eq:ineqEta2}
1-\etaop\{\psi_A\}\leq \sqrt{[1+(d_A-1)\etaop\{\cC\}][1-\eta_P\{\psi_A\}]}.
\end{equation}
Inverting Eq.~\eqref{eq:ineqEta2} and recalling the definitions of $\etaop$ and $\eta_P$ as the maximization of $\eta_{(\cdot)}\{\psi_A\}$ over all states in the code yields the right inequality of Eq.~\eqref{eq:ineqEta}.

Equation~\eqref{eq:ineqEta} implies that an approximately correctable subspace code must necessarily be such that the fidelity loss for the transpose channel is small.
A small fidelity loss for the transpose channel in turn requires that $\cE$ has Kraus operators that satisfy Eq.~\eqref{eq:AOQECcond} with $\Delta_{ij}$ small.
Equation.~\eqref{eq:AOQECcond} with $\Delta_{ij}$ small is thus not only sufficient (as shown in Sec.~\ref{sec:sufficiency}), but also necessary for subspace codes.

An obvious extension of the current case to subsystem codes is one where $\cE$ is a product channel, that is, $\cE(\rho_A\otimes\rho_B)=\cF_A(\rho_A)\otimes \cF_B(\rho_B)$, for CPTP (on their respective domains) channels $\cF_A$ and $\cF_B$.
For such an $\cE$, the transpose channel is also a product channel, namely, the product of the respective transpose channels of $\cF_A$ and $\cF_B$.
Since there is no flow of information between $A$ and $B$, whether subsystem $A$ is correctable relies only on the properties of $\cF_A$.
We can thus treat this case as if we have a subspace code on $A$ under noise $\cF_A$, for which the transpose channel is immediately near-optimal from Theorem~\ref{thm:subspace}.

\subsection{Maximally mixed state on subsystem $B$}\label{sec:maxmixed}

Consider the subset of code states where $B$ is in the maximally mixed state,
\begin{equation}\label{eq:C0}
\cC_0\equiv\left\{\rho_A\otimes \frac{P_B}{d_B},\quad\forall\rho_A\in\cS(\cH_A)\right\}\subset\cC.
\end{equation}
For states in $\cC_0$, the action of the noise channel $\cE$ can be written as
\begin{equation}
\cE{\left(\rho_A\otimes \frac{P_B}{d_B}\right)}=\sum_{is}\bar E_{is}\rho_A \bar E_{is}^\dagger\equiv \bar\cE_A(\rho_A).
\end{equation}
$\bar\cE_A$ is a CPTP channel on $A$ with Kraus operators $\{\bar E_{is}\equiv (1/\sqrt{d_B})E_i|s_B\rangle\}$, where $\{|s_B\rangle\}_{s=1}^{d_B}$ is an orthonormal basis for $\cH_B$.
Let us, for a moment, forget about subsystem $B$ and ask about correctability of $\cC_0$---now viewed as a subspace code on $A$---under the noise $\bar\cE_A$.
Theorem~\ref{thm:subspace} applies and ensures that the transpose channel corresponding to noise $\bar\cE_A$ and code $\cC_{0}$, denoted as $\cR_{A,P}$, has fidelity loss close to that of the optimal recovery $\cR_{A,\textrm{op}}$,
\begin{equation}\label{eq:maxmixed}
\eta_P\{\cC_0\}\leq \etaop\{\cC_0\} f(\etaop\{\cC_0\};d_A),
\end{equation}
for $f(\eta;d)$ defined in Eq.~\eqref{eq:f}.

Such a code $\cC_0$ is of practical relevance whenever one lacks control over subsystem $B$. Full control over subsystem $A$ alone is sufficient to guarantee preparation of a product code state, while rapid and complete decoherence (for example) causes the state on $B$ to quickly approach a random state well-described by the maximally mixed state. Equation \eqref{eq:maxmixed} reassures us that, in this case, the transpose channel still works well as a recovery map.

We are, however, more interested in the performance of the transpose channel on the original subsystem code $\cC$ where the state on $B$ is unrestricted. After all, the freedom to choose the state of $B$, without incurring adverse effects on the information-carrying capability of the code, is the essence of a subsystem code.
Observe that, for any state in $\cC_0$, the actions of $(\tr_B\circ\RP\circ\cE)$ (as usual, $\RP$ is the transpose channel for $\cE$ on code $\cC$) and $(\cR_{A,P}\circ\bar\cE_A)$ are identical.
The optimal recovery for $\cC_0$, however, need not be the same map as the optimal recovery for $\cC$---the optimal recovery for $\cC$ has to work well for \emph{all} states in $\cC$, not just those in $\cC_0$.
However, since $\cC_0\subset\cC$, we have the following inequality,
\begin{equation}\label{eq:C0C}
\etaop\{\cC_0\}\leq \etaop\{\cC\}.
\end{equation}
Furthermore, since $\eta f(\eta;d)$ is a monotonically increasing function of $\eta$, we can combine \eqref{eq:maxmixed} and \eqref{eq:C0C} to obtain the following corollary,
\begin{corollary}\label{cor:maxmixed}
For the subset of code states $\cC_0 \subset \cC$,
\begin{equation}
\eta_P\{\cC_0\}\leq \etaop\{\cC\}f(\etaop\{\cC\};d_A).\label{eq:maxmixed2}
\end{equation}
\end{corollary}
\noindent This says that for the code states where $B$ is in the maximally mixed state---which can be viewed as the ``average state" for the full degree of freedom described by $B$---the transpose channel works nearly as well as the optimal recovery operation for $\cC$.

As an aside, we note that a ``state-dependent" transpose channel is near-optimal for codes where $B$ is always prepared in some \emph{known} state.
Code $\cC_0$ is a special case of this, but now $B$ can be in a fixed state other than the maximally mixed state.
For example, the rapid decoherence process of subsystem $B$ may have a fixed point that is not the maximally mixed state (for instance, the ground state, if the noise is dissipative), so that any initial preparation of the $B$ state quickly relaxes into this fixed state.
Such codes should properly be viewed as (isomorphic to) subspace codes on subsystem $A$.
Since the identity of the state on $B$ is known, the optimal recovery map for this code must make use of this knowledge.
Likewise, the associated transpose channel, in order to work well, must also depend on the fixed state on $B$.
Using similar techniques as above, one can show
\begin{equation}
\eta_{P_{\phi_B}}\{\cC_{\phi_B}\}\leq \etaop\{\cC_{\phi_B}\} f(\etaop\{\cC_{\phi_B}\};d_A),
\end{equation}
where $\cC_{\phi_B}$ is the set of code states with $\phi_B$ as the fixed state on $B$, and $\eta_{P_{\phi_B}}$ refers to the fidelity loss of the state-dependent transpose channel with Kraus operators $ \{ (P_A \otimes \sqrt{\phi_{B}})E_{i}^\dagger \left[\cE(P_A\otimes\phi_B)\right]^{-1/2}\}$. This is similar to previously known results from Ref.~\cite{BK}, derived in the context of entanglement fidelity for reversing dynamics on a given input state.

\subsection{$B$ is perfectly correctable}\label{sec:Bcorr}
Suppose subsystem $B$ is in fact perfectly correctable, but we choose to use subsystem $A$ to store the information.
A simple and often-encountered example is where the noise on $B$ is describable by Kraus operators that are products of Pauli operators.
More generally, any noise process that satisfies the perfect QEC conditions for $B$ falls under our current considerations.
Despite the perfect correctability on $B$, one might still choose to store information in $A$, for example, when $B$ is experimentally inaccessible or uncontrollable, or if $A$ is a much larger system with greater storage capacity than $B$.
The transpose channel is again near-optimal in this case.

We demonstrate this near-optimality by first showing that, for $B$ perfectly correctable, the fidelity for a pure initial state on subsystem $A$ from using the transpose channel as recovery is independent of the initial state of subsystem $B$.
\begin{lemma}\label{lem}
If subsystem $B$ is perfectly correctable under noise $\cE$, then $F{\left[|\psi\rangle_A,(\tr_B\circ\RP\circ\cE)(\psi_A\otimes\rho_B)\right]}$, where $\psi_A\equiv |\psi_A\rangle\langle\psi_A|$, is independent of $\rho_B$.
\end{lemma}
\begin{proof}
$B$ perfectly correctable under noise $\cE$ implies that the perfect QEC conditions (Eq.~\eqref{cond1} of Theorem~\ref{thm:equivcond} with the roles of $A$ and $B$ interchanged) hold: There exists operators $A_{ij}$ on $A$ for all $i,j$ such that $PE_i^\dagger \cE(P)^{-1/2}E_j P=A_{ij}\otimes P_B$. From this, we have $F^2\left[|\psi_A\rangle,\left(\tr_B\circ\RP\circ\cE\right)(\psi_A\otimes\rho_B)\right]=\sum_{ij}|\langle\psi_A|A_{ij}|\psi_A\rangle|^2$,
which is independent of $\rho_B$.
\end{proof}

Lemma~\ref{lem} implies the following sequence of relations:
\begin{align}
\eta_P\{\cC\}=\max_{\rho\in\cC}\eta_P\{\rho\}&=\max_{\rho=\psi_A\otimes\rho_B}\eta_P\{\rho\}\nonumber\\
&=\max_{\psi_A}\eta_P\{\psi_A\otimes P_B/d_B\}\nonumber\\
&=\eta_P\{\cC_0\}\nonumber\\
&\leq \etaop\{\cC\}f(\etaop\{\cC\};d_A).\label{eq:Bperf}
\end{align}
The second equality in the first line of Eq.~\eqref{eq:Bperf} follows from the concavity of the fidelity, with $\psi_A$ denoting a pure state. The second line makes use of Lemma~\ref{lem}, and the last inequality is just Eq.~\eqref{eq:maxmixed2}.
Equation~\eqref{eq:Bperf} gives exactly the right inequality in Eq.~\eqref{eq:ineqEta} applied to the current scenario, from which we draw the conclusion that the transpose channel is near-optimal on $A$ under channel $\cE$ when $B$ is perfectly correctable.

\subsection{$\cE$ destroys distinguishability on $B$}\label{sec:arbitrarystate}
 Suppose the noise process $\cE$ satisfies the following condition:
 \begin{condition}\label{condTr}
For CPTP $\cE$, suppose there exists $\delta\geq 0$ such that
\begin{equation}
\left\Vert \cE(\rho_A\otimes \rho_B)-\cE{\left(\rho_A\otimes \frac{P_B}{d_B}\right)}\right\Vert_\text{tr}\leq \delta{\left\Vert\rho_B-\frac{P_B}{d_B}\right\Vert}_\text{tr}
\end{equation}
for all states $\rho_A \in \cS(\cH_A)$ and $\rho_B \in \cS(\cH_B)$. $\Vert O\Vert_\tr$ denotes the trace norm of $O$ given by $\tr|O|$.
\end{condition}
\noindent If $\delta\ll 1$, any two states on $\cH_B$, after the action of $\cE$, become close together and nearly indistinguishable (as quantified by the trace norm used in Condition~\ref{condTr}) from each other.
A simple example is a product channel $\cE=\cE_A\otimes \cE_B$ where $\cE_B$ maps all states on $B$  to some fixed state $\tau_B$. In this case, $\delta$ can be chosen to be zero.
While we have chosen, for convenience of subsequent analysis, to state Condition~\ref{condTr} in terms of comparing states on $B$ before and after the channel $\cE$ to what happens to the maximally mixed state $P_B/d_B$, one is free to choose other reference states on $B$ if desired.

For channels and codes satisfying Condition~\ref{condTr}, the transpose channel also works well as a recovery channel, as encapsulated in the following corollary:
\begin{corollary}\label{thm:AbStateFidelity}
Given that Condition~\ref{condTr} is satisfied, for a subsystem code $\cC$,
\begin{align}
\eta_P&\leq (d_A+1)\etaop+3\delta+O(\delta^2,\etaop^2,\etaop\delta).
\end{align}
\end{corollary}
\noindent The proof of this corollary is detailed in Appendix~\ref{app}. The idea behind the proof is to first show that the transpose channel works well as a recovery for the information stored in $A$ when $B$ is initially in the maximally mixed state. Since Condition~\ref{condTr} says that $\cE$ brings code states with different states on $B$ close together, if the transpose channel works well as a recovery for $B$ being initially in the maximally mixed state, it will also work well when $B$ is initially in a different state.

Corollary~\ref{thm:AbStateFidelity}, like similar statements before, tells us that the fidelity loss for the transpose channel is not much worse than that of the optimal recovery. The additional fidelity loss suffered from using the simpler transpose channel rather than the optimal recovery is governed by $d_A$ as well as the parameter $\delta$ which characterizes how badly $\cE$ destroys distinguishability between states on subsystem $B$.

\section{Conclusion}\label{sec:concl}

We studied the role of the transpose channel in approximate quantum error correction.
We first obtained a set of conditions for perfect subsystem error correction that explicitly involves the transpose channel. This completes our understanding as to why certain channels admit perfectly correctable codes, in a particularly intuitive way. Our perfect QEC conditions naturally lead to sufficient conditions for approximate QEC, where the resilience to noise of the information stored in the code is quantified in a simple way. We also demonstrated that the transpose channel works nearly as well as any other recovery channel for four different scenarios of codes and noise. In all these cases, the near-optimality of the transpose channel relies only on $d_A$, the dimension of the information-carrying subsystem $A$, and not on $d_B$, the dimension of the noisy subsystem that carries no information.

Using our transpose channel approach to derive necessary conditions for approximate QEC for general subsystem codes will provide the final missing link in our unifying and analytical framework for understanding approximate quantum error correction.
Even disproving our conjecture that the transpose channel is a universally good recovery operation for approximate codes will be a useful step forward.
In this case, then, the question will be to discover a different recovery map that can serve as a universal recovery.

Another possible extension is to consider codes that include not just product states on $AB$ in $\cC$ (as we have done), but also correlated states.
Once there is correlation between $A$ and $B$, it is, of course, no longer clear where the information initially resides.
If one has complete control over the preparation of the initial code states, it would be simpler to make use of the subsystem structure and confine the information to only one subsystem.
Practically, however, experimental restrictions may result in an initial (possibly small) correlation between $A$ and $B$, leading to a different notion of ``approximate" or imperfection in the code.
Such a situation was previously studied in \cite{Shabani05} in the context of perfectly noiseless subsystems that require no careful initialization.
One can ask similar questions for approximate subsystem codes.

A separate future direction is to perform the transpose channel recovery on experimental implementations of approximate codes. The transpose channel, like any CPTP map, can be implemented physically using operations on an extended Hilbert space. The more pertinent and fruitful question, however, will be to discover simple and efficient ways of implementing the transpose channel on a specific physical system of our choice.

\section{Acknowledgments}
P.M. would like to thank David Poulin, John Preskill and Todd Brun for useful discussions. H.K.N is supported by the National Research Foundation and the Ministry of Education, Singapore.

\appendix

\section{Proof of Theorem \ref{thm:equivcond}}\label{app:perfectCond}
In this section, we prove Theorem \ref{thm:equivcond}, by demonstrating the equivalence between Eq.~\eqref{cond1} (statement (\ref{eq:A}) below) and the perfect QEC conditions (statement (\ref{eq:B}) below) derived in \cite{NielsenPoulin07, OQECC, OQEClong}.

\begin{lemma}\label{lem:equivcond}
Given a CP channel $\cE : \cB(\cH_{AB}) \rightarrow \cB(\cP_\cE)$ with a set of Kraus operators $\{E_{i}\}$, the following two statements are equivalent:
\begin{enumerate}[(A)]
\item\label{eq:A} $PE_i^\dagger \cE(P)^{-1/2}E_jP=P_A\otimes B_{ij}$, for all ${i,j}$, and $B_{ij}\in\cB(\cH_B)$;
\item\label{eq:B} $PE_i^\dagger E_jP=P_A\otimes B'_{ij}$, for all ${i,j}$, and $B'_{ij}\in\cB(\cH_B)$.
\end{enumerate}
\end{lemma}

\begin{proof}
(\ref{eq:A})$\Rightarrow$(\ref{eq:B}): For any $i,j$, $\sum_k(P_A\otimes B_{ik})(P_A\otimes B_{kj})=\sum_k(PE_i^\dagger\cE(P)^{-1/2}E_kP)(PE_k^\dagger\cE(P)^{-1/2}E_jP)=PE_i^\dagger E_jP$, which gives $PE_i^\dagger E_jP=P_A\otimes B'_{ij}$, with $B'_{ij}\equiv\sum_kB_{ik}B_{kj}$.

(\ref{eq:B})$\Rightarrow$(\ref{eq:A}): Let $\{|s\rangle_B\}$ be an orthonormal basis for $\cH_B$. (\ref{eq:B}) implies
\begin{equation}\label{eq:2to1a}
P_AE_{is}^\dagger E_{jt}P_A=\lambda_{(is)(jt)}P_A,
\end{equation}
where $E_{is}\equiv E_i|s\rangle_B$ is an operator that brings vectors in $\cH_A$ to vectors in $\cH_{AB}$, and $\lambda_{(is)(jt)}\equiv\langle s|B'_{ij}|t\rangle$. $\{E_{is}\}$ is a set of Kraus operators for the CP channel $\cE_A:\cB(\cH_A)\rightarrow \cB(\cP_\cE)$ defined by $\cE_A(\rho_A)=\cE(\rho_A\otimes P_B)$. We view $\Lambda\equiv {\left(\lambda_{(is)(jt)}\right)}$ as a two-index matrix, where the first index is the double index $(is)$, and the second is $(jt)$. Observe that $\lambda_{(jt)(is)}^*=\lambda_{(is)(jt)}$, i.e., $\Lambda$ is a hermitian matrix. It is thus diagonalizable, i.e., $\exists~U\equiv {\left(u_{(is)(jt)}\right)}$ such that $U\Lambda U^\dagger = \Lambda_D$, where $\Lambda_D$ is a diagonal matrix. More explicitly, we have
\begin{equation}
\sum_{(i's'),(j't')}u_{(is)(i's')}\lambda_{(i's')(j't')}u^*_{(jt)(j't')}=\delta_{(is)(jt)}d_{is},
\end{equation}
where $d_{is}$ are the diagonal entries of $\Lambda_D$. Using this, we can write Eq.~\eqref{eq:2to1a} in its diagonal form:
\begin{equation}\label{eq:diag2}
P_AF^\dagger_{is}F_{jt}P_A=\delta_{ij}\delta_{st}d_{is}P_A,
\end{equation}
where $F_{is}\equiv \sum_{(i's')}u_{(is)(i's')}^*E_{i's'}$ gives a different Kraus representation for $\cE_A$. Equation~\eqref{eq:diag2} gives the polar decomposition $F_{is}P_A=\sqrt{d_{is}}V_{is}P_A$, where $V_{is}$ is a unitary operator satisfying $P_AV^\dagger_{is}V_{jt}P_A=\delta_{ij}\delta_{st}P_A$. Let $P_{is}\equiv V_{is}P_AV_{is}^\dagger$. Then, $P_{is}$'s are orthogonal projectors, since $P_{is}P_{jt}=\delta_{ij}\delta_{st}P_{is}$. Direct computation gives $\cE(P)=\sum_{is}d_{is}P_{is}$, i.e., $\cE(P)$ is a sum of orthogonal projectors, and hence easy to invert: $\cE(P)^{-1/2}=\sum_{is}d_{is}^{-1/2}P_{is}$. Further algebra gives $P E_i^\dagger \cE(P)^{-1/2} E_j P=P_A\otimes B_{ij}$, with $B_{ij}\equiv\sum_{st}\sum_{kv}u_{(kv)(is)}^* u_{(kv)(jt)}\sqrt{d_{kv}}|s\rangle_B \langle t|$.
\end{proof}

References~\cite{NielsenPoulin07, OQECC, OQEClong} showed that a code $\cC$ is perfectly correctable on $A$ under $\cC$ if and only if statement $\eqref{eq:B}$ is true. This fact, together with Lemma \ref{lem:equivcond}, proves Theorem \ref{thm:equivcond}.

\section{Proof of Corollary~\ref{thm:AbStateFidelity}}\label{app}

To prove Corollary~\ref{thm:AbStateFidelity}, we first need the following lemma:

\begin{lemma}\label{thm:nearoptimalC}
Consider a subsystem code $\cC$ under noise $\cE$. For any pure state $\psi_A\equiv |\psi_A\rangle\langle\psi_A|$,
\begin{align}
&\quad 1-\etaop{\left\{\psi_A\otimes \frac{P_B}{d_B}\right\}}\nonumber\\
&\leq \sqrt{[1+(d_A-1)\etaop\{\cC\}]{\left[1-\eta_P{\left\{\psi_A\otimes \frac{P_B}{d_B}\right\}}\right]}}.\label{eq:appmaxmix}
\end{align}
\end{lemma}
\noindent The proof of this lemma proceeds exactly as in the proof used to demonstrate Eq.~\eqref{eq:ineqEta2} (see \cite{aqecPRA}), except for the minor modification that $B$ is a nontrivial subsystem.

With this lemma, we can prove Corollary~\ref{thm:AbStateFidelity}:
\smallskip

\noindent{\bf Corollary~\ref{thm:AbStateFidelity}.} {\it Given that Condition~\ref{condTr} is satisfied, for a subsystem code $\cC$,}
\begin{align}
\eta_P&\leq (d_A+1)\etaop+3\delta+O(\delta^2,\etaop^2,\etaop\delta).
\end{align}

\begin{proof}
For channel $\cE$ satisfying Condition~\ref{condTr}, for any recovery $\cR$ and any state $\psi_A\otimes\rho_B\equiv |\psi_A\rangle\langle\psi_A|\otimes \rho_B\in\cC$,
\begin{align}
&\quad F^2\left[|\psi_A\rangle,\tr_B\left\{(\cR\circ\cE)(\psi_A\otimes\rho_B)\right\}\right]\nonumber\\
&\leq\delta+F^2\left[|\psi\rangle_A,\tr_B\left\{(\cR \circ\cE)(\psi_A\otimes P_B/d_B)\right\}\right],\label{eqL}
\end{align}
which implies
\begin{align}
\eta_\cR\{\psi_A\otimes \rho_B\}\geq \eta_\cR\{\psi_A\otimes P_B/d_B\}-\delta.\label{eqH}
\end{align}
Interchanging the roles of $\rho_B$ and $P_B/d_B$ in Eq.~\eqref{eqL} yields, similarly,
\begin{align}
\eta_\cR\{\psi_A\otimes P_B/d_B\}\geq \eta_{\cR}\{\psi_A\otimes \rho_B\}-\delta.\label{eqI}
\end{align}
We have the following sequence of inequalities:
\begin{align}
&\quad\etaop\{C\}+\delta\nonumber\\
&\geq \etaop\{\psi_A\otimes \rho_B\}+\delta\nonumber\\
&\geq \etaop\{\psi_A\otimes P_B/d_B\}\quad(\text{using Eq. }\eqref{eqH})\nonumber\\
&\geq 1-\sqrt{\big[1+(d_A-1)\etaop\{\cC\}\big]\left(1-\eta_P{\left\{\psi_A\otimes \frac{P_B}{d_B}\right\}}\right)}\nonumber\\
&\hspace{4.5cm}(\text{using Eq. }\eqref{eq:appmaxmix})\nonumber\\
&\geq 1-\sqrt{[1+(d_A-1)\etaop\{\cC\}](1-\eta_P\{\psi_A\otimes \rho_B\}+\delta)}\nonumber\\
&\hspace{4.5cm}(\text{using Eq. }\eqref{eqI})\label{eqJ}
\end{align}
Rearranging Eq.~\eqref{eqJ} and recalling that $\eta_P=\max_{\rho\in\cC}\eta_P\{\rho\}$ immediately gives the statement of the corollary.
\end{proof}

\end{document}